\documentclass{LMCS}

\def\doi{9(1:05)2013}
\lmcsheading%
{\doi}
{1--16}
{}
{}
{Mar.~29, 2012}
{Feb.~27, 2013}
{}

\usepackage[latin1]{inputenc}
\usepackage[T1]{fontenc}
\usepackage{hyperref,enumerate}

\usepackage{amssymb}
\usepackage{tikz}
\usetikzlibrary{arrows,decorations.pathmorphing,backgrounds,positioning,fit}

\usepackage{amssymb,amsmath}

\usepackage[english]{babel}

\usepackage{verbatim}

\usepackage{url}

\usepackage{galois}

\usepackage{tikz,pgf}

  \newcommand{\setN}{\mathbb{N}}
  \newcommand{\setZ}{\mathbb{Z}}
  \newcommand{\setQ}{\mathbb{Q}}
  
  \newcommand{\norm}[2]{||#2||_{#1}}
  \newcommand{\moins}{\backslash}

\theoremstyle{plain}

\theoremstyle{definition}

\theoremstyle{remark}

\renewcommand{\vec}[1]{{\mathbf #1}}

\begin{document}

  \title[Vector Addition System Reversible Reachability Problem]{Vector Addition System Reversible Reachability Problem\rsuper*} 

\author[J.~Leroux]{Jérôme Leroux}	
\thanks{Work funded by ANR grant REACHARD-ANR-11-BS02-001.}
\address{LaBRI, Univ. Bordeaux, CNRS}	
\email{leroux@labri.fr}

\keywords{Vector addition system, reachability, boundedness, cover.}
\ACMCCS{[{\bf Theory of computation}]:  Logic---Logic and verification}
\subjclass{F.3.1}
\amsclass{68R99, 68Q05, 03D99.}
\titlecomment{{\lsuper*}Work based on the earlier extended abstracts \cite{DBLP:conf/concur/Leroux11}.}

 \maketitle
  
 \begin{abstract}
   The reachability problem for vector addition systems is a central problem of net theory. This problem is known to be decidable but the complexity is still unknown. Whereas the problem is EXPSPACE-hard, no elementary upper bounds complexity are known. In this paper we consider the reversible reachability problem. This problem consists to decide if two configurations are reachable one from each other, or equivalently if they are in the same strongly connected component of the reachability graph. We show that this problem is EXPSPACE-complete. As an application of the introduced materials we characterize the reversibility domains of a vector addition system. 
 \end{abstract}

\section{Introduction}
Vector addition systems (VASs) or equivalently Petri nets are one of the most popular formal methods \cite{survey-esparza} for the representation and the analysis of parallel processes. Their reachability problem is central since many computational problems (even outside the realm of parallel processes) reduce to the reachability problem. Sacerdote and Tenney provided in \cite{sacerdote77} a partial proof of decidability of this problem. The proof was completed in 1981 by Mayr \cite{Mayr81} and simplified by Kosaraju \cite{Kosaraju82} from \cite{sacerdote77,Mayr81}. Ten years later \cite{lambert-structure}, Lambert provided a further simplified version based on \cite{Kosaraju82}. This last proof still remains difficult and the upper-bound complexity of the corresponding algorithm is just known to be non-primitive recursive. Nowadays, the exact complexity of the reachability problem for VASs is still an open-problem. The problem is known to be EXPSPACE-hard \cite{Lipton76}. Note that the existence of a primitive recursive upper bound of complexity for the reachability problem is still open since the Zakaria Bouziane's paper~\cite{743436} introducing such a bound was proved to be incorrect by Petr Jan\v{c}ar~\cite{Jancar:2008:BTP:1453257.1453392}.

\smallskip

Recently, in \cite{EasyChair:31}, a new proof of the reachability problem based on the notion of \emph{transformer relations} inspired by Hauschildt~\cite{Hauschildt:90} was published. That proof shows that reachability sets are \emph{almost semilinear}, a class of sets introduced in that paper that extends the class of Presburger sets. An application of that result was provided; a final configuration is proved to be not reachable from an initial one if and only if there exists a forward inductive invariant definable in the Presburger arithmetic that contains the initial configuration but not the final one. Since we can decide if a Presburger formula denotes a forward inductive invariant, we deduce that there exist checkable certificates of non-reachability in the Presburger arithmetic. In particular, there exists a simple algorithm for deciding the general VAS reachability problem based on two semi-algorithms. A first one that tries to prove the reachability by enumerating finite sequences of actions and a second one that tries to prove the non-reachability by enumerating Presburger formulas. The Presburger inductive invariants presented in that paper is obtained thanks to strongly connected subreachability graphs (called \emph{witness graph} and recalled in Section~\ref{sec:rev}). As a direct consequence, configurations in these graphs are reachable one from each other.

\smallskip

In this paper we consider the reversible reachability problem that consists to decide if two configurations are reachable one from each other. We prove that this problem is EXPSPACE-complete. This result extends known result for the subclasses of reversible and cyclic vector addition systems~\cite{Bouziane97,Lipton76}. We also prove that the general \emph{coverability problem} reduces to the reversible reachability problem (see Section~\ref{sec:vas}). As an application of the introduced materials we characterize the reversibility domains of a vector addition system in the last Section~\ref{sec:application}.

\section{Projected Vectors}
We denote by $\setZ$ and $\setN$ the set of \emph{integers} and \emph{natural numbers}. In this paper, some components of vectors in $\setZ^d$ are projected away. In order to avoid multiple dimensions, we introduce an additional element $\star\not\in\setZ$, the set $\setZ_\star=\setZ\cup\{\star\}$, and the set $\setZ_I^d$ of vectors $\vec{z}\in\setZ_\star^d$ such that $I=\{i \mid \vec{z}(i)=\star\}$. Operations on $\setZ$ are extended component-wise into operations on $\setZ_I^d$ by interpreting $\star$ as a projected component. More formally we denote by $\vec{z}_1+\vec{z}_2$ where $\vec{z}_1,\vec{z}_2\in\setZ_I^d$ the vector $\vec{z}\in\setZ_I^d$ defined by $\vec{z}(i)=\vec{z}_1(i)+\vec{z}_2(i)$ for every $i\not\in I$. Symmetrically given $\vec{z}\in\setZ_I^d$ and an integer $k\in\setZ$, we denote by $k\vec{z}$ the vector in $\setZ_I^d$ defined by $(k\vec{z})(i)=k(\vec{z}(i))$ for every $i\not\in I$. The usual order $\leq$ over $\setZ$ is extended over $\setZ_\star$ into the unique total order $\leq$ satisfying $z\leq \star$ for every $z\in\setZ_\star$. The relation $\leq$ is extended component-wise over $\setZ_*^d$.

\begin{exa}
  We have $k(\star,1)=(\star,k)$ even if $k=0$. We also have $(\star,5)-(\star,2)=(\star,3)$ and $(\star,1)+(\star,2)=(\star,3)$. We have $\cdots\leq -1\leq 0\leq 1\leq \cdots \leq \star$.
\end{exa}

The \emph{projection} of a vector $\vec{z}\in\setZ_I^d$ by eliminating components indexed by $L\subseteq\{1,\ldots,d\}$ is the vector in $\setZ_{I\cup L}^d$ defined by $\pi_L(\vec{z})(i)=\vec{z}(i)$ for every $i\not\in L$. The projection of a set $\vec{Z}\subseteq \setZ_I^d$  by eliminating components indexed by $L$ is defined as expected by $\pi_L(\vec{Z})=\{\pi_L(\vec{z})\mid \vec{z}\in\vec{Z}\}$.

\begin{exa}
  Let $L=\{1\}$. We have $\pi_L(1000,1)=(\star,1)$ and $\pi_L(4,\star)=(\star,\star)$. We also have $\pi_L(\{(2,0),(1,1),(2,0)\})=\{(\star,0),(\star,1),(\star,2)\}$.
\end{exa}

Let $\vec{z}\in\setZ_I^d$. We denote by $\norm{\infty}{\vec{z}}$ the natural number equals to $0$ if $I=\{1,\ldots,d\}$ and equals to $\max_{i \not\in  I}|\vec{z}(i)|$ otherwise. Given a finite set $\vec{Z}\subseteq \setZ_I^d$ we denote by $\norm{\infty}{\vec{Z}}$ the natural number $\max_{\vec{z}\in\vec{Z}}\norm{\infty}{\vec{z}}$ if $\vec{Z}$ is non empty and $0$ if $\vec{Z}$ is empty.

\section{Vector Addition Systems}\label{sec:vas}
A \emph{Vector Addition System} (\emph{VAS}) is a finite set $\vec{A}\subseteq \setZ^d$. Vectors $\vec{a}\in \vec{A}$ are called \emph{actions} and vectors $\vec{c}\in\setN_\star^d$ with $\setN_\star=\setN\cup\{\star\}$ are called \emph{configurations}. A configuration in $\setN^d$ is said to be \emph{standard} and we denote by $\setN_I^d$ the set of configurations $\vec{c}\in\setN_\star^d$ such that $I=\{i\mid \vec{c}(i)=\star\}$. Given a word $\sigma=\vec{a}_1\ldots\vec{a}_k$ of actions $\vec{a}_j\in \vec{A}$ we denote by $\Delta(\sigma)$ the vector in $\setZ^d$ defined by $\Delta(\sigma)=\sum_{j=1}^k\vec{a}_j$. This vector is called the \emph{displacement} of $\sigma$. We also introduce the vector $\Delta_I(\sigma)=\pi_I(\Delta(\sigma))$. A \emph{run} $\rho$ from a configuration $\vec{x}\in\setN_I^d$ to a configuration $\vec{y}\in\setN_I^d$ \emph{labelled} by a word $\sigma=\vec{a}_1\ldots\vec{a}_k$ of actions $\vec{a}_j\in\vec{A}$ is a non-empty word $\rho=\vec{c}_0\ldots\vec{c}_k$ of configurations $\vec{c}_j\in\setN_I^d$ such that $\vec{c}_0=\vec{x}$, $\vec{c}_k=\vec{y}$ and such that $\vec{c}_j=\vec{c}_{j-1}+\pi_I(\vec{a}_j)$ for every $j\in\{1,\ldots,k\}$. Note that in this case $\rho$ is unique and $\vec{y}-\vec{x}=\Delta_I(\sigma)$. This run is denoted by $\vec{x}\xrightarrow{\sigma}\vec{y}$. The set $I$ is called the set of \emph{projected components} of $\rho$. The \emph{projection} $\pi_L(\rho)$ of a run $\rho=\vec{c}_0\ldots\vec{c}_k$  by eliminating components indexed by $L\subseteq\{1,\ldots,d\}$ is defined as expected as the run $\pi_L(\rho)=\pi_L(\vec{c}_0)\ldots\pi_L(\vec{c}_k)$. Observe that if $\rho$ is the run $\vec{x}\xrightarrow{\sigma}\vec{y}$ then $\pi_L(\rho)$ is the run $\pi_L(\vec{x})\xrightarrow{\sigma}\pi_L(\vec{y})$. The following lemma provides a simple way to deduce a converse result. 
\begin{lem}\label{lem:holding}
  Let $L$ be a set of indexes and $\vec{c}$ be a configuration such that there exists a run from $\pi_L(\vec{c})$ labelled by a word $\sigma$. If $\vec{c}(i)\geq |\sigma|~\norm{\infty}{\vec{A}}$ for every $i\in L$ then there exists a run from $\vec{c}$ labelled by $\sigma$.
\end{lem}
\begin{proof}
  Let $\vec{c}\in\setN_I^d$ be a configuration such that there exists a path from $\pi_L(\vec{c})$ labelled by a word $\sigma=\vec{a}_1\ldots\vec{a}_k$ where $\vec{a}_j\in\vec{A}$. Let us introduce the vector $\vec{c}_j=\vec{c}+\pi_I(\vec{a}_1+\ldots+\vec{a}_j)$. Since there exists a run from $\pi_L(\vec{c})$ labelled by $\sigma$ we deduce that $\pi_L(\vec{c}_j)\in \setN_{I\cup L}^d$. Observe that for every $j\in\{0,\ldots,k\}$ and for every $i\not\in I$ we have $\vec{c}_j(i)\geq \vec{c}(i)-|\sigma|~\norm{\infty}{\vec{A}}$. In particular if $\vec{c}(i)\geq |\sigma|~\norm{\infty}{\vec{A}}$ for every $i\in L\moins I$ we deduce that $\vec{c}_j\in\setN_I^d$. Therefore $\rho=\vec{c}_0\ldots\vec{c}_k$ is the run from $\vec{c}$ labelled by $\sigma$.
\end{proof}

\begin{exa}
  $\rho=(2,0)(1,1)(0,2)$ is the run $(2,0)\xrightarrow{(-1,1)(-1,1)}(0,2)$. Let $L=\{1\}$ and observe that $\pi_L(\rho)=(\star,0)(\star,1)(\star,2)$ is the run $(\star,0)\xrightarrow{(-1,1)(-1,1)}(\star,2)$.
\end{exa}

Let $\vec{x}$ and $\vec{y}$ be two standard configurations. When there exists a run from $\vec{x}$ to $\vec{y}$ we say that $\vec{y}$ is \emph{reachable} from $\vec{x}$ and if there also exists a run from $\vec{y}$ to $\vec{x}$ we say that $(\vec{x},\vec{y})$ is in the \emph{reversible reachability relation}. The problem of deciding this last property is called the \emph{reversible reachability problem}. This problem is shown to be EXPSPACE-hard by introducing the \emph{coverability problem}. Given two standard configurations $\vec{x}$ and $\vec{y}$ we say that $\vec{y}$ is \emph{coverable} by $\vec{x}$ if there exists a standard configuration in  $\vec{y}+\setN^d$ reachable from $\vec{x}$. The coverability problem is known to be EXPSPACE-complete \cite{Lipton76,rackoff78}. By reducing the coverability problem to the reversible reachability problem we get the following lemma.
\begin{lem}
  The reversible reachability problem is EXPSPACE-hard.
\end{lem}
\begin{proof}
  We consider a vector addition system $\vec{A}$. We first observe that we can add to the vector addition system $\vec{A}$ additional actions of the form $(0,\ldots,0,-1,0,\ldots,0)$ without modifying the coverability problem. Thanks to this transformation a standard configuration $\vec{y}$ is coverable from a standard configuration $\vec{x}$ if and only if $\vec{y}$ is reachable from $\vec{x}$. We introduce the VAS $\vec{V}$ in dimension $d+2$ defined by $\vec{V}=((0,0)\times \vec{A})\cup \{(-1,1,-\vec{y}),(1,-1,\vec{x})\}$.

  Let us prove that $(1,0,\vec{x})$ and $(0,1,\vec{0})$ are in the reversible reachability relation of $\vec{V}$ if and only if $\vec{y}$ is coverable from $\vec{x}$ in $\vec{A}$. In fact if $\vec{y}$ is coverable from $\vec{x}$ in $\vec{A}$, then $\vec{y}$ is reachable from $\vec{x}$ thanks to the additional actions $(0,\ldots,0,-1,0,\ldots,0)$. Hence there exists a run from $\vec{x}$ to $\vec{y}$ labelled by a word $\vec{a}_1\ldots\vec{a}_k$ of actions $\vec{a}_j\in\vec{A}$. The following runs shows that $(1,0,\vec{x})$ and $(0,1,\vec{0})$ are in the reversible reachability relation of $\vec{V}$:
  $$(1,0,\vec{x})\xrightarrow{(0,0,\vec{a}_1)\ldots(0,0,\vec{a}_k)}(1,0,\vec{y})\xrightarrow{(-1,1,-\vec{y})}(0,1,\vec{0})\xrightarrow{(1,-1,\vec{x})}(1,0,\vec{x})$$
  Converselly, let us assume that $(1,0,\vec{x})$ and $(0,1,\vec{0})$ are in the reversible reachability relation of $\vec{V}$. Hence there exists a run from $(1,0,\vec{x})$ to $(0,1,\vec{0})$ labelled by a word $\sigma$. We consider the maximal prefix $w$ of $\sigma$ in $((0,0)\times\vec{A})^*$. This word has the special form $w=(0,0,\vec{a}_1)\ldots(0,0,\vec{a}_k)$. Observe that $w$ is the label of a run from $(1,0,\vec{x})$ to a standard configuration of the form $(1,0,\vec{z})$. We deduce that $\vec{a}_1\ldots\vec{a}_k$ is a the label of run in $\vec{A}$ from $\vec{x}$ to $\vec{z}$. Moreover, since $(1,0,\vec{z})\not=(0,1,\vec{0})$ we deduce that $w$ is a strict prefix of $\sigma$. Let $\vec{v}\in\vec{V}$ such that $w\vec{v}$ is a prefix of $\sigma$. By maximality of $\sigma$ we deduce that $\vec{v}\in  \{(-1,1,-\vec{y}),(1,-1,\vec{x})\}$. Since $(1,0,\vec{z})+\vec{v}\geq \vec{0}$, we get $\vec{v}=(-1,1,-\vec{y})$. Thus $\vec{z}\geq \vec{y}$ and we have proved that $\vec{y}$ is coverable from $\vec{x}$ in $\vec{A}$.
  
  As a direct consequence, the reversible reachability problem is EXPSPACE-hard.
\end{proof}

\section{Subreachability Graphs}
A \emph{subreachability graph} is a graph $G=(\vec{Q},T)$ where $\vec{Q}\subseteq \setN_I^d$ is a non empty finite set of configurations called \emph{states} and $T\subseteq \vec{Q}\times\vec{A}\times\vec{Q}$ is a finite set of triples $(\vec{x},\vec{a},\vec{y})\in\vec{Q}\times \vec{A}\times \vec{Q}$ satisfying  $\vec{x}\xrightarrow{\vec{a}}\vec{y}$ called \emph{transitions}. The set $I$ is called the set of \emph{projected components} of $G$ and the subreachability graph is said to be \emph{standard} if $I$ is empty. A \emph{witness graph} is a strongly connected subreachability graph (see Fig.~\ref{fig:subgraph} for examples). The \emph{projection} $\pi_L(t)$ of a transition $t=(\vec{x},\vec{a},\vec{y})$  by eliminating components indexed by $L\subseteq \{1,\ldots,d\}$ is defined by $\pi_L(t)=(\pi_L(\vec{x}),\vec{a},\pi_L(\vec{y}))$ and the projection of the set of transitions $T$ is defined by $\pi_L(T)=\{\pi_L(t)\mid t\in T\}$. The projection $\pi_L(G)$ of a subreachability graph $G=(\vec{Q},T)$ is the subreachability graph $\pi_L(G)=(\pi_L(\vec{Q}),\pi_L(T))$.

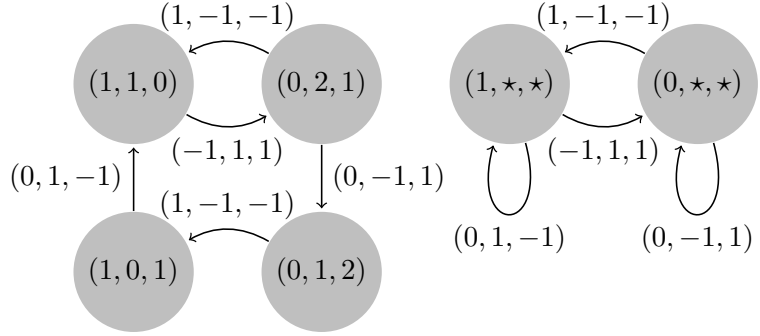
\begin{figure}
  \begin{center}
\begin{tikzpicture}[->,shorten >=1pt,auto,node distance=2.5cm,on grid,semithick, state/.style={circle,fill=black!25}]
    
    \node[state] (Aup) {$(1,1,0)$};
    \node[state] (Bup) [right=of Aup] {$(0,2,1)$};
    \node[state] (Adown) [below=of Aup] {$(1,0,1)$};
    \node[state] (Bdown) [right=of Adown] {$(0,1,2)$};

    \path (Aup) edge [bend right] node[anchor=north] {$(-1,1,1)$} (Bup);
    \path (Bup) edge [bend right] node[anchor=south] {$(1,-1,-1)$} (Aup);
    \path (Bdown) edge [bend right]  node[anchor=south] {$(1,-1,-1)$} (Adown);
    
    \path (Bup) edge  node[right] {$(0,-1,1)$} (Bdown);
    \path (Adown) edge node[left] {$(0,1,-1)$} (Aup);

    \node[state] (A) [right=of Bup]{$(1,\star,\star)$};
    \node[state] (B) [right=of A] {$(0,\star,\star)$};   
    \path (A) edge [bend right] node[anchor=north] {$(-1,1,1)$} (B);
    \path (B) edge [bend right] node[anchor=south] {$(1,-1,-1)$} (A);
    \path (B) edge [loop below] node{$(0,-1,1)$} (B);
    \path (A) edge [loop below] node{$(0,1,-1)$} (A);
  \end{tikzpicture}

  \end{center}
  \caption{A subreachability graph $G$ and the subreachability graph $\pi_L(G)$ with $L=\{2,3\}$.\label{fig:subgraph}}
\end{figure}

\begin{exa}\label{ex:subreachability}
  A standard subreachability graph $G=(\vec{Q},T)$ and the subreachability graph $\pi_L(G)$ obtained from $G$ by eliminating components indexed by $L=\{2,3\}$ are depicted in Fig.~\ref{fig:subgraph}.
\end{exa}

A \emph{path} in a subreachability graph $G$ from a configuration $\vec{x}\in\vec{Q}$ to a configuration $\vec{y}\in\vec{Q}$ labelled by a word $\sigma=\vec{a}_1\ldots\vec{a}_k$ of actions $\vec{a}_j\in\vec{A}$ is a word $p=t_1\ldots t_k$ of transitions $t_j\in T$ of the form $t_j=(\vec{c}_{j-1},\vec{a}_j,\vec{c}_j)$ with $\vec{c}_0=\vec{x}$ and $\vec{c}_k=\vec{y}$. We observe that the word $p$ is unique. This path is denoted by $\vec{x}\xrightarrow{\sigma}_G\vec{y}$. Let us observe that in this case $\rho=\vec{c}_0\ldots\vec{c}_k$ is the unique run $\vec{x}\xrightarrow{\sigma}\vec{y}$. In particular if a path $\vec{x}\xrightarrow{\sigma}_G\vec{y}$ exists then the run $\vec{x}\xrightarrow{\sigma}\vec{y}$ also exists. Note that conversely if there exists a run $\vec{x}\xrightarrow{\sigma}\vec{y}$ then there exists a subreachability $G$ such that $\vec{x}\xrightarrow{\sigma}_G\vec{y}$. Such a $G$ is obtained by introducing the set of states $\vec{Q}=\{\vec{c}_0,\ldots,\vec{c}_k\}$ and the set of transitions $T=\{t_1,\ldots,t_k\}$ where $t_j=(\vec{c}_{j-1},\vec{a}_j,\vec{c}_j)$. A path $\vec{x}\xrightarrow{\sigma}_G\vec{y}$ is called a \emph{cycle} if $\vec{x}=\vec{y}$. The cycle is said to be \emph{simple} if $\vec{c}_{j_1}=\vec{c}_{j_2}$ with $j_1<j_2$ implies $j_1=0$ and $j_2=k$. The projection $\pi_L(p)$ of a path $p=t_1\ldots t_k$ in $G$  by eliminating components indexed by $L\subseteq \{1,\ldots,d\}$ is the path $\pi_L(p)=\pi_L(t_1)\ldots\pi_L(t_k)$ in $\pi_L(G)$. Observe that the projection of a path $\vec{x}\xrightarrow{\sigma}_G\vec{y}$  by eliminating components indexed by $L$ is the path $\pi_L(\vec{x})\xrightarrow{\sigma}_{\pi_L(G)}\pi_L(\vec{y})$. The \emph{Parikh image} of a path is the function $\mu:T\rightarrow\setN$ defined by $\mu(t)$ is the number of occurrences of $t$ in this path. A cycle is said to be \emph{total} if its Parikh image $\mu$ satisfies $\mu(t)\geq 1$ for every $t\in T$.

\begin{exa}
  Let us come back to the standard witness graph $G$ depicted in Fig.~\ref{fig:subgraph}. Let us consider the cycle $(1,1,0)\xrightarrow{(-1,1,1)(1,-1,-1)}_G(1,1,0)$ in $G$. Its projection  by eliminating components indexed by $L=\{2,3\}$ is the cycle $(1,\star,\star)\xrightarrow{(-1,1,1)(1,-1,-1)}_{\pi_L(G)}(1,\star,\star)$ in the witness graph $\pi_L(G)$ also depicted in Fig.~\ref{fig:subgraph}.
\end{exa}

A word $\sigma\in\vec{A}^*$ is said to be \emph{forward iterable} from a configuration $\vec{c}$ if there exists a run $\vec{c}\xrightarrow{\sigma}\vec{y}$ such that $\vec{c}\leq \vec{y}$. In this case the configuration $\vec{c}_\star=\pi_L(\vec{c})$ where $L=\{i \mid \vec{c}(i)\not=\vec{y}(i)\}$ is called the \emph{forward limit} of $\sigma$ from $\vec{c}$. We observe that $\sigma$ is forward iterable from $\vec{c}$ if and only if for every $n\in\setN$ there exists a run $\vec{c}\xrightarrow{\sigma^n}\vec{y}_n$. In that case $L$ is the minimal set of indexes such that $\pi_L(\vec{y}_n)$ does not depend on $n$. Symmetrically $\sigma$ is said to be \emph{backward iterable} from a configuration $\vec{c}$ if there exists a run $\vec{x}\xrightarrow{\sigma}\vec{c}$ such that $\vec{c}\leq \vec{x}$. In this case the configuration $\vec{c}_\star=\pi_L(\vec{c})$  where $L=\{i \mid \vec{c}(i)\not=\vec{x}(i)\}$ is called the \emph{backward limit} of $\sigma$ from $\vec{c}$.

\begin{exa}
  The action $\vec{a}=(0,-1,1)$ is forward iterable from $\vec{x}=(0,\star,0)$ since $(0,\star,0)\xrightarrow{\vec{a}}(0,\star,1)$. Observe that in this case $(0,\star,0)\xrightarrow{\vec{a}^n}(0,\star,n)$ for every $n\in\setN$. The forward limit of $\vec{a}$ from $(0,\star,0)$ is $(0,\star,\star)$.
\end{exa}

A configuration $\vec{c}$ is said to be \emph{forward pumpable} by a cycle $\vec{q}\xrightarrow{\sigma}_G\vec{q}$ if $\sigma$ is forward iterable from $\vec{c}$ with a forward limit equals to $\vec{q}$. Note that in this case $\vec{q}$ is unique since it satisfies $\vec{q}=\pi_I(\vec{c})$ where $I$ is the set of projected components of $G$. Symmetrically a configuration $\vec{c}$ is said to be \emph{backward pumpable} by a cycle $\vec{q}\xrightarrow{\sigma}_G\vec{q}$ if $\sigma$ is backward iterable from $\vec{c}$ with a backward limit equals to $\vec{q}$.

\begin{exa}
  Let us come back to the witness graph $\pi_L(G)$ depicted in Fig.~\ref{fig:subgraph}. Observe that $(0,\star,0)$ is forward pumpable by  $(0,\star,\star)\xrightarrow{(0,-1,1)}_{\pi_L(G)}(0,\star,\star)$.
\end{exa}

\section{Outline}
The remainder of this paper is a proof that the reversible reachability problem is in EXPSPACE. We prove that if a pair $(\vec{x},\vec{y})$ of standard configurations are in the reversible reachability relation then there exist runs from $\vec{x}$ to $\vec{y}$ and from $\vec{y}$ to $\vec{x}$ with lengths bounded by a number double exponential in the size of $(\vec{x},\vec{A},\vec{y})$ with the binary encoding for numbers. Using the fact that NEXPSPACE=EXPSPACE, and that double exponential numbers can be stored
in exponential space, one obtain the EXPSPACE upper bound. These ``short'' runs are obtained as follows.

Theorem~\ref{thm:bound} gives a bound on the size of the Parikh image of a cycle in a witness graph to achieve a particular displacement vector,
using a result of Pottier~\cite{P-RTA91}. This result is used in Section~\ref{sec:rev}, which considers the special case of reversible witness graphs in which each path can be followed by another path such that the total displacement is zero.  In Theorem~\ref{thm:reversible} it is shown that a reversible witness graph possesses a ``short'' total cycle that has a zero displacement.

Section~\ref{sec:pump} takes an arbitrary witness graph $G$ and asserts the existence of a set of indexes $J$ such that the witness graph $\pi_J(G)$ has a ``small'' number of states and such that states $\vec{q}$ of $G$ that are not ``too'' large are forward and backward pumpable by ``short'' cycles in $\pi_J(G)$.

The development culminates with the main result in Section~\ref{sec:deciding}.  There, we consider a reversible witness graph where $\vec{x}$ and $\vec{y}$ are two states. This graph is finite but potentially very large. One then uses the result
from Section~\ref{sec:pump} to generate a reversible witness graph $\pi_J(G)$ satisfying the previous conditions in such a way that $\vec{x}$ and $\vec{y}$ can be considered as not ``too'' large configurations. Most of the work involves showing how to replace arbitrary path between $\vec{x}$ and $\vec{y}$ by ``short'' paths by exploiting the fact that $\vec{x}$ and $\vec{y}$ are pumpable to move from $\pi_J(G)$ back to $G$.

\section{Displacement Vectors}\label{sec:euler}
A \emph{displacement vector}  of a witness graph $G$ is a finite sum of vectors of the form $\Delta(\sigma)=\sum_{j=1}^k\vec{a}_j$ where $\sigma=\vec{a}_1\ldots\vec{a}_k$ is a word labelling a cycle in $G$. We denote by $\vec{Z}_G$ the set of \emph{displacement vectors}. Observe that $\vec{Z}_G$ is a submonoid of $(\setZ^d,+)$. Displacement vectors are related to \emph{Kirchhoff functions} as follows. A \emph{Kirchhoff function} for a witness graph $G=(\vec{Q},T)$ is a function $\mu:T\rightarrow\setN$ such that the functions $\operatorname{in}(\mu),\operatorname{out}(\mu):\vec{Q}\rightarrow\setN$ defined bellow are equal.
$$
\operatorname{in}(\mu)(\vec{x})=\sum_{t \in T  \cap (\vec{Q}\times \vec{A}\times \{\vec{x}\})}\mu(t)
~~~~~~~
\operatorname{out}(\mu)(\vec{x})=\sum_{t \in T \cap (\{\vec{x}\}\times \vec{A}\times \vec{Q})}\mu(t) 
$$

A Kirchhoff function $\mu:T\rightarrow \setN$ is said to be \emph{total} if $\mu(t)\geq 1$ for every $t\in T$.
\begin{lem}[Euler's Lemma]
  A function $\mu$ is a Kirchhoff function for a witness graph $G$ if and only if $\mu$ is a finite sum of Parikh images of cycles in $G$. In particular a function $\mu$ is a total Kirchhoff function if and only if $\mu$ is the Parikh image of a  total cycle.
\end{lem}
As a direct consequence of the Euler's Lemma, we deduce that a vector $\vec{z}\in\setZ^d$ is a displacement vector of $G$ if and only if there exists a Kirchhoff function $\mu$ for $G$ satisfying the following equality:
$$\vec{z}=\sum_{t=(\vec{x},\vec{a},\vec{y})\in T}\mu(t)\vec{a}$$
In this case $\vec{z}$ is called the \emph{displacement} of $\mu$.

\begin{exa}\label{ex:group}
  Let us come back to the witness graph $\pi_L(G)$ depicted in Fig.~\ref{fig:subgraph}. A function $\mu:\pi_L(T)\rightarrow\setN$ is a Kirchhoff function for $\pi_L(G)$ if and only if $\mu(t_1)=\mu(t_2)$ where $t_1=((1,\star,\star),(-1,1,1),(0,\star,\star))$ and $t_2=((0,\star,\star),(1,-1,-1),(1,\star,\star))$. In particular the set of displacement vectors of $\pi_L(G)$ satisfies $\vec{Z}_{\pi_L(G)}=\{\vec{z}\in\setZ^3\mid \vec{z}(1)=0\wedge\vec{z}(2)+\vec{z}(3)=0\}$.
\end{exa}

The following theorem shows that the displacement vectors $\vec{z}\in\vec{Z}_G$ are displacement of Kirchhoff functions $\mu$ for $G$ such that $\norm{\infty}{\mu}=\max_{t\in T}\mu(t)$ is bounded by a polynomial in $|\vec{Q}|$, $\norm{\infty}{\vec{A}}$, and $\norm{\infty}{\vec{z}}$ with a degree depending on $d$.
\begin{thm}\label{thm:bound}
  Vectors $\vec{z}\in\vec{Z}_G$ are displacement of Kirchhoff functions $\mu$ such that the following inequality holds where $q=|\vec{Q}|$, $a=\norm{\infty}{\vec{A}}$, and $m=\norm{\infty}{\vec{z}}$:
  $$\norm{\infty}{\mu}\leq 
  (q^{d+1} a(1+2a)^d+m)^d$$
\end{thm}
\begin{proof}
  We first recall a ``Frobenius  theorem'' proved in \cite{P-RTA91}. Let $H\in\setZ^{d\times n}$ be a matrix and let us denote by $h_{i,j}$ for each $i\in\{1,\ldots,d\}$ and $j\in\{1,\ldots,n\}$ the element of $H$ at position $(i,j)$. We denote by $\norm{1,\infty}{H}$ the natural number $\max_{1\leq i\leq d}\sum_{j=1}^n|h_{i,j}|$. Given a vector $\vec{v}\in\setN^n$, we introduce the natural number $\norm{1}{\vec{v}}=\sum_{j=1}^n\vec{v}(j)$. Let $\vec{V}$ be the set of vectors $\vec{v}\in\setN^n$ such that $H\vec{v}=\vec{0}$. Recall that $\vec{V}$ is a submonoid of $(\setN^n,+)$ generated by the finite set $\min(\vec{V}\moins\{\vec{0}\})$ of minimal elements for $\leq$. From \cite{P-RTA91} we deduce that vectors $\vec{v}\in \min(\vec{V}\moins\{\vec{0}\})$ satisfy the following inequality where $r$ is the rank of $H$:
  $$\norm{1}{\vec{v}}\leq (1+\norm{1,\infty}{H})^r$$
  
  \medskip
  
  Observe that if $a=0$ then $\vec{z}=\vec{0}$ and the theorem is proved with the Kirchhoff function $\mu$ defined by $\mu(t)=0$ for every $t\in T$. So we can assume that $a\geq 1$. Since every cycle labelled by a word $\sigma$ can be decomposed into a finite sequence of simple cycles labelled by words $\sigma_1,\ldots,\sigma_k$ such that $\Delta(\sigma)=\sum_{j=1}^k\Delta(\sigma_j)$ we deduce that the set of displacement vectors $\vec{Z}_G$ is the submonoid of $(\setZ^d,+)$ generated by the set $\vec{Z}$ of non-zero vectors $\vec{z}=\Delta(\sigma)$ where $\sigma$ is the label of a simple cycle. Since the length of a simple cycle is bounded by the cardinal $q$ of $\vec{Q}$, we get $\norm{\infty}{\vec{Z}}\leq q a$. As a corollary we deduce that the cardinal $k$ of $\vec{Z}$ is bounded by $k\leq (1+2q a )^d-1$ (the $-1$ comes from the fact that vectors in $\vec{Z}$ are non-zero). 

\medskip

Let us consider a vector $\vec{z}\in\vec{Z}_G$ and let us introduce a whole enumeration $\vec{z}_1,\ldots,\vec{z}_k$ of the vectors in $\vec{Z}$ and the following set $\vec{V}$ where $n=k+1$:
$$\vec{V}=\{\vec{v}\in\setN^n \mid \bigwedge_{i=1}^d\sum_{j=1}^k\vec{v}(j)\vec{z}_j(i)-\vec{v}(n)\vec{z}(i)=0\}$$
We observe that $\vec{V}$ is associated to a matrix $H\in\setZ^{d\times n}$. The rank of $H$ is bounded by $d$ and $\norm{1,\infty}{H}\leq k q a+m$. We deduce from the Frobenius theorem that vectors $\vec{v}\in \min(\vec{V}\moins\{\vec{0}\})$ satisfy the following inequality:
$$\norm{1}{\vec{v}}\leq (1+k q a + m)^d\leq (q^{d+1} a (1 +2 a)^d +m)^d$$
Since $\vec{z}\in\vec{Z}_G$ and $\vec{Z}_G$ is the submonoid generated by $(\setZ^d,+)$ generated by $\vec{z}_1,\ldots,\vec{z}_k$, we deduce that there exists $v_1,\ldots,v_k\in\setN$ such that $\vec{z}=\sum_{j=1}^kv_j\vec{z}_j$. Observe that the vector $\vec{v}\in\setN^n$ defined by $\vec{v}(j)=v_j$ if $j\in\{1,\ldots,k\}$ and $\vec{v}(n)=1$ is in $\vec{V}$. Hence, there exists $\vec{v}\in\vec{V}$ such that $\vec{v}(n)=1$. In particular there exists another vector $\vec{v}\in \min(\vec{V}\moins\{\vec{0}\})$ such that $\vec{v}(n)=1$. Observe that for every $j\in\{1,\ldots,k\}$ there exists a function $\lambda_j$ that is the Parikh image of a simple cycle such that $\vec{z}_j$ is the displacement of $\lambda_j$. We introduce the Kirchhoff function $\mu=\sum_{j=1}^k\vec{v}(j)\lambda_j$. Since $\vec{v}\in\vec{V}$ and $\vec{v}(n)=1$ we deduce that the displacement of $\mu$ is $\vec{z}$. The theorem is proved by observing that
$\mu(t)
= \sum_{j=1}^k\vec{v}(j) \lambda_j(t)
\leq \norm{1}{\vec{v}}$ since $\lambda_j(t)\in\{0,1\}$.
\end{proof}

\section{Reversible Witness Graphs}\label{sec:rev}
A witness graph $G$ is said to be \emph{reversible} if for every path $\vec{x}\xrightarrow{u}_G\vec{y}$ there exists a path $\vec{y}\xrightarrow{v}_G\vec{x}$ such that $\Delta(u) +\Delta(v)=\vec{0}$. Observe that standard witness graphs are reversible since the condition $\Delta(u) +\Delta(v)=\vec{0}$ is implied by the two paths.

\begin{exa}
  The witness graphs depicted in Fig.~\ref{fig:subgraph} are
  reversible, but the witness graph $(\{\star\},\{(\star,1,\star)\})$ is not.
\end{exa}

Let us recall that a submonoid $\vec{Z}$ of $(\setZ^d,+)$ is said to be a \emph{subgroup} if $-\vec{z}\in\vec{Z}$ for every $\vec{z}\in\vec{Z}$. The following lemma provides two characterizations of the reversible witness graphs.
\begin{lem}\label{lem:rev}
  A witness graph $G$ is reversible if and only if $\vec{Z}_G$ is a subgroup of $(\setZ^d,+)$ if and only if the zero vector is the displacement of a total Kirchhoff function.
\end{lem}
\begin{proof}
  Assume first that $G$ is reversible and let us prove that $\vec{Z}_G$ is a subgroup of $(\setZ^d,+)$. Let us consider a cycle $\vec{x}\xrightarrow{u}_G\vec{x}$. Since $G$ is reversible, there exists a cycle $\vec{x}\xrightarrow{v}_G\vec{x}$ such that $\Delta(u)+\Delta(v)=\vec{0}$. We deduce that $-\vec{Z}_G=\vec{Z}_G$ since vectors in $\vec{Z}_G$ are finite sums of vectors $\Delta(u)$ where $u$ is the label of a cycle in $G$. Therefore $\vec{Z}_G$ is a subgroup of $\setZ^d$.

  Now let us assume that $\vec{Z}_G$ is a subgroup of $(\setZ^d,+)$ and let us prove that the zero vector is the displacement of a total Kirchhoff function. Since $G$ is strongly connected, there exists a total cycle $\vec{x}\xrightarrow{u}_G\vec{x}$. Observe that $\vec{z}=\Delta(u)$ is in $\vec{Z}_G$. Since $\vec{Z}_G$ is a subgroup we deduce that $-\vec{z}\in\vec{Z}_G$. Hence $-\vec{z}$ is the displacement of a Kirchhoff function $\lambda$. Let $\lambda'$ be Parikh image of  $\vec{x}\xrightarrow{u}_G\vec{x}$ and observe that $\mu=\lambda+\lambda'$ is a total Kirchhoff function. Moreover the displacement of $\mu$ is $-\vec{z}+\vec{z}=\vec{0}$.

  Finally, let us assume that the zero vector is the displacement of a total Kirchhoff function $\mu$ and let us prove that $G$ is reversible. Let us consider a path $\vec{x}\xrightarrow{u}_G\vec{y}$. Since $G$ is strongly connected, there exists a path $\vec{y}\xrightarrow{\alpha}_G\vec{x}$. Let us consider the Parikh image $\lambda$ of the cycle $\vec{x}\xrightarrow{u \alpha}_G\vec{x}$ and let $m=1+\norm{\infty}{\lambda}$. We observe that $\mu'=m\mu-\lambda$ is a total Kirchhoff function and the Euler's Lemma shows that $\mu'$ is the Parikh image of a cycle $\vec{x}\xrightarrow{\beta}_G\vec{x}$. From $\mu'=m\mu-\lambda$ we deduce that $\Delta(\beta)=m\vec{0}-\Delta(u\alpha)$. Let us consider $v=\alpha\beta$ and observe that $\vec{y}\xrightarrow{v}_G\vec{x}$ and $\Delta(u)+\Delta(v)=\vec{0}$. Thus $G$ is reversible.
\end{proof}

The following theorem shows that if $G$ is a reversible witness graph then the zero vector is the displacement of a total Kirchhoff function $\mu$ such $\norm{\infty}{\mu}$ can be bounded by a polynomial in $|\vec{Q}|$ and $\norm{\infty}{\vec{A}}$ with a degree depending on $d$.
\begin{thm}\label{thm:reversible}
  Let $G$ be a reversible witness graph. The zero vector is the displacement of a total Kirchhoff function $\mu$ such that the following inequality holds where $q=|\vec{Q}|$ and $a=\norm{\infty}{\vec{A}}$:
  $$\norm{\infty}{\mu}\leq (q(1+2a))^{d(d+1)}$$
\end{thm}
\proof
  Since $G$ is strongly connected, every transition $t\in T$ occurs in at least one simple cycle. We denote by $\lambda_t$ the Parikh image of such a simple cycle and we introduce the Kirchhoff function $\lambda=\sum_{t\in T}\lambda_t$. We have $\lambda(t)\in\{1,\ldots,|T|\}$ for every $t\in T$. We introduce the displacement $\vec{z}$ of $\lambda$. Since $G$ is reversible, we deduce that $-\vec{z}$ is the displacement vector of a Kirchhoff function for $G$ by Lemma~\ref{lem:rev}. As $\norm{\infty}{\vec{z}}\leq |T| q a$, $|T|\leq q |\vec{A}|$, and $|\vec{A}|\leq (1+2a)^d$ we deduce that $\norm{\infty}{\vec{z}}\leq q^2a(1+2a)^d$. Theorem~\ref{thm:bound} shows that $-\vec{z}$ is the displacement of a Kirchhoff function $\lambda'$ satisfying the following inequalities:
  $$
  \norm{\infty}{\lambda'}
  \leq (q^{d+1} a(1+2 a)^d +q^2 a (1+2 a )^d)^d 
  \leq (q^{d+1}2a(1+2a)^d)^d
  $$
  Let us consider the total Kirchhoff function $\mu=\lambda+\lambda'$. Observe that the displacement of $\mu$ is the zero vector and since $\norm{\infty}{\lambda}\leq |T|\leq q(1+2 a)^d\leq (q^{d+1}(1+2a)^d)^d$ we get the theorem with:
\[
  \norm{\infty}{\mu}
  \leq (q^{d+1}2a(1+2a)^d)^d +  (q^{d+1}(1+2a)^d)^d
  \leq (q(1+2a))^{d(d+1)}\eqno{\qEd}
\]

\section{Extractors}\label{sec:relax}
In this section we introduce a way for extracting ``large'' components of configurations. An \emph{extractor} is a non increasing sequence $\lambda=(\lambda_n)_{1\leq n\leq d}$ of natural numbers $\lambda_n\in\setN$. Let $\vec{X}\subseteq\setN_I^d$. An \emph{excluding} set for $(\lambda,\vec{X})$ is a set of indexes $J$ such that $\vec{x}(i)<\lambda_{|J|+1}$ for every $i\not\in J$ and for every $\vec{x}\in\vec{X}$ (notice that even if $\lambda_{d+1}$ is not defined, when $|J|=d$ the domain of the universal quantifier ``for every $i\not\in J$'' is empty). Since $\lambda$ is non increasing we deduce that the class of excluding sets for a couple $(\lambda,\vec{X})$ is stable by intersection. As this class contains $\{1,\ldots,d\}$ we deduce that there exists a \emph{unique minimal excluding set} $J$ for $(\lambda,\vec{X})$. By minimality of this set we deduce that for every $i\in J$ there exists $\vec{x}\in\vec{X}$ such that $\vec{x}(i)\geq \lambda_{|J|}$ (notice once again that even if $\lambda_0$ is not defined, when $|J|=0$ the domain of the universal quantifier ``for every $i\in J$'' is empty). We denote $\lambda(\vec{X})$ the set $\pi_J(\vec{X})$ where $J$ is the minimal excluding set for $(\lambda,\vec{X})$.

\begin{exa}
  Let $\lambda=(5,3,2)$ be an extractor. We have $\lambda(\{(1,8,1)\})=\{(1,\star,1)\}$, and $\lambda(\{(1,8,1),(3,1,1)\})=\{(\star,\star,1)\}$.
\end{exa}

A set $\vec{X}\subseteq \setN_I^d$ is said to be \emph{normalized} for $\lambda$ if $\lambda(\vec{X})=\vec{X}$. As a direct consequence of the following lemma we deduce that $\lambda(\vec{X})$ is normalized for $\lambda$ for every set $\vec{X}\subseteq \setN_I^d$. We say that $\vec{x}\in\setN_I^d$ is \emph{normalized} for $\lambda$ if $\{\vec{x}\}$ is normalized for $\lambda$, i.e $\lambda(\{\vec{x}\})=\{\vec{x}\}$ or equivalently $\vec{x}(i)<\lambda_{|I|+1}$ for every $i\not\in I$. Observe that if every state $\vec{x}\in\vec{X}$ is normalized then $\lambda(\vec{X})=\vec{X}$.
\begin{lem}\label{lem:excluding}
  Let $\vec{X}\subseteq \setN_I^d$ and let $L$ be a set of indexes included in the minimal excluding set of $(\lambda,\vec{X})$. Then $\lambda(\vec{X})=\lambda(\pi_L(\vec{X}))$.
\end{lem}
\begin{proof}
  Note that if $\vec{X}$ is empty the result is immediate so we can assume that $\vec{X}$ is non empty. Let $J$ be the minimal excluding set of $(\lambda,\vec{X})$ and observe that $J$ is an excluding set for $\vec{X}'=\pi_L(\vec{X})$. In particular the minimal excluding set $J'$ for $\vec{X}'$ satisfies $J'\subseteq J$. Since $J'$ is an excluding set of $(\lambda,\vec{X}')$ we deduce that $\vec{x}'(i)<\lambda_{|J'|+1}$ for every $i\not\in J'$. Hence $\pi_L(\vec{x})(i)<\lambda_{|J'|+1}$ for every $\vec{x}\in\vec{X}$. As $\vec{x}\leq\pi_L(\vec{x})$ we deduce that $J'$ is an excluding set of $(\lambda,\vec{X})$. By minimality of $J$ we get the other inclusion $J\subseteq J'$. Thus $J=J'$ and we have proved that $\lambda(\vec{X})=\lambda(\pi_L(\vec{X}))$.
\end{proof}

\section{Pumpable Configurations}\label{sec:pump}
In this section we show that for arbitrary witness graph $G$, there exists a set $J$ of indexes such that the number of states of $\pi_J(G)$ is ``small'' and such that states with ``small'' size of $G$ are pumpable by ``short'' cycles of $\pi_J(G)$. The proof of this result is inspired by the Rackoff ideas \cite{rackoff78}. All other results or definitions introduced in this section are not used in the sequel.

\begin{thm}\label{thm:pumpable}
  Let $G$ be a witness graph with a set of states $\vec{Q}\subseteq \setN_I^d$, and let $s\in\setN_{>0}$ be a positive integer. We introduce the positive integer $x=(1+\norm{\infty}{\vec{A}})s$. There exists a set of indexes $J$ such that the number of states of $\pi_J(G)$ is bounded by $x^{d^d}$ and such that every state $\vec{q}\in\vec{Q}$ such that $\norm{\infty}{\vec{q}}<s$ is forward and backward pumpable by cycles of $\pi_J(G)$ with lengths bounded by $d x^{d^d}$.
\end{thm}

Such a set $J$ is obtained by introducing the class of \emph{adapted extractors}. An extractor $\lambda$ is said to be \emph{adapted} if the following inequality holds for every $n\in\{2,\ldots,d\}$:
$$\lambda_{n-1}\geq \lambda_{n}^{d-n+1}\norm{\infty}{\vec{A}} + \lambda_{n}$$

\begin{lem}\label{lem:forward}
  Let $\lambda$ be an adapted extractor, $G$ be a witness graph with a set of states $\vec{Q}\subseteq \setN_I^d$, and let $J$ be the minimal excluding set for $(\lambda,\vec{Q})$. For every state $\vec{q}\in\vec{Q}$ there exists a run $\vec{q}\xrightarrow{u}\vec{y}$ such that $\pi_J(\vec{q})\xrightarrow{u}_{\pi_J(G)}\pi_J(\vec{y})$ and such that the bounds $|u|\leq \sum_{|I|<n\leq |J|}\lambda_n^{d+1-n}$, and $\vec{y}(j)\geq \lambda_{|J|}$ for every $j\in J$ hold.
\end{lem}
\begin{proof}
  Since $\vec{Q}\subseteq \setN_I^d$ we deduce that $I\subseteq J$. We introduce a parameter $k\in\setN$ and we prove the lemma by induction over $k$ under the constraint $|J|-|I|\leq k$. Observe that if $k=0$ then $I=J$ and the property is proved with $u=\epsilon$ and $\vec{y}=\vec{q}$. Assume the property proved for a natural number $k\in\setN$ and let us consider a witness graph $G=(\vec{Q},T)$ with a set of projected components $I$ such that $|J|-|I|\leq k+1$ where $J$ is the minimal excluding set for $(\lambda,\vec{Q})$. We consider a state $\vec{q}\in\vec{Q}$. If $\vec{Q}$ is normalized for $\lambda$ then $J=I$ and the property is proved. So we can assume that $\vec{Q}$ is not normalized for $\lambda$. We deduce that there exists a state in $\vec{Q}$ that is not normalized. Since $G$ is strongly connected, there exists a path $\vec{q}\xrightarrow{\sigma}_G\vec{p}$ with a minimal length such that $\vec{p}$ is not normalized. Let us observe that the number of states in $\vec{Q}$ that are normalized is bounded by $\lambda_{|I|+1}^{d-|I|}$. By minimality of the length of $\sigma$ we deduce that $|\sigma|\leq \lambda_{|I|+1}^{d-|I|}$. 
  
  We introduce the minimal excluding set $K$ for $(\lambda,\{\vec{p}\})$. Observe that $I$ is strictly included in $K$ since $\vec{p}$ is not normalized. Moreover $K$ is included in $J$ since $J$ is an excluding set for $(\lambda,\{\vec{p}\})$. Lemma~\ref{lem:excluding} shows that $J$ is the minimal excluding set of $(\lambda,\pi_K(\vec{Q}))$. Observe that $|J|-|K|<|J|-|I|\leq k+1$. By applying the induction on the witness graph $\pi_K(G)$ and the state $\pi_K(\vec{p})$, we deduce that there exists a run 
  $\pi_K(\vec{p})\xrightarrow{u}\vec{y}$ such that $\pi_J(\vec{p})\xrightarrow{u}_{\pi_J(G)}\pi_J(\vec{y})$ with $|u|\leq \sum_{|K|<n\leq |J|}\lambda_n^{d+1-n}$ and such that $\vec{y}(j)\geq \lambda_{|J|}$ for every $j\in J$. We introduce the word $v=\sigma u$. Since $I$ is strictly included in $K$ we deduce that $\sum_{|I|<n\leq |K|}\lambda_n^{d+1-n}\geq \lambda_{|I|+1}^{d-|I|}$. Thus $|v|\leq \sum_{|I|<n\leq |J|}\lambda_n^{d+1-n}$.
  
  Since $\lambda$ is an adapted extractor we deduce that $\lambda_{|K|}\geq \norm{\infty}{\vec{A}}\sum_{|K|<n\leq |J|}\lambda_n^{d+1-n}+\lambda_{|J|}$. From $\vec{p}(k)\geq \lambda_{|K|}$ for every $k\in K$ we deduce that $\vec{p}(k)\geq \norm{\infty}{\vec{A}}|u|+\lambda_{|J|}$.  Since there exists a run from $\pi_K(\vec{p})$ labelled by $u$, Lemma~\ref{lem:holding} shows that there exists a run $\vec{p}\xrightarrow{u}\vec{z}$. For every $k\in K$ we have $\vec{z}(k)\geq \vec{p}(k)-\norm{\infty}{\vec{A}}|u|\geq \lambda_{|J|}$. As $\vec{p}\xrightarrow{u}\vec{z}$ we deduce that $\pi_K(\vec{p})\xrightarrow{u}\pi_K(\vec{z})$. In particular $\pi_K(\vec{z})=\vec{y}$. Let $j\in J\moins K$. From the previous equality we get $\vec{z}(j)=\vec{y}(j)$. Moreover since $\vec{y}(j)\geq \lambda_{|J|}$ we get $\vec{z}(j)\geq \lambda_{|J|}$. We have proved that $\vec{z}(j)\geq \lambda_{|J|}$ for every $j\in J$. Hence the induction is proved.
\end{proof}

Now let us prove Theorem~\ref{thm:pumpable}. We consider a witness graph $G$ with a set of states $\vec{Q}\subseteq\setN_I^d$. We also consider a positive integer $s\in\setN_{>0}$ and we introduce the positive integers $a=\norm{\infty}{\vec{A}}$ and $x=(1+a)s$. Let $\lambda$ be the adapted extractor defined by $\lambda_d=s$ and the following induction for every $n\in\{2,\ldots,d\}$:
$$\lambda_{n-1}=\lambda_{n}^d(1+\norm{\infty}{\vec{A}})$$
An immediate induction provides $\lambda_n^{d+1-n}\leq x^{d^d}$ for every $n\in\{1,\ldots,d\}$. We introduce the minimal excluding set $J$ for $(\lambda,\vec{Q})$. Observe that if $|J|=d$ then $|\pi_J(\vec{Q})|=1$ and in particular $|\pi_J(\vec{Q})|\leq x^{d^d}$. If $|J|<d$, the number of states in $\pi_J(\vec{Q})$ is bounded by $\lambda_{|J|+1}^{d-|J|}$. Hence $|\pi_J(\vec{Q})|\leq x^{d^d}$ in any case. Let us consider $\vec{q}\in\vec{Q}$ such that $\norm{\infty}{\vec{q}}<s$. Lemma~\ref{lem:forward} shows that there exists a run $\vec{q}\xrightarrow{\sigma}\vec{x}$ with $\vec{x}(j)\geq \lambda_{|J|}$ for every $j\in J$ such that $\pi_J(\vec{q})\xrightarrow{\sigma}_{\pi_J(G)}\pi_J(\vec{x})$ and such that:
$$|\sigma|\leq \sum_{n=1}^{|J|}\lambda_n^{d+1-n}$$
Since $\pi_J(G)$ is strongly connected there exists a path $\pi_J(\vec{x})\xrightarrow{u}_{\pi_J(G)}\pi_J(\vec{q})$. We can assume that the length of $u$ is minimal. In particular $u=\epsilon$ if $J=\{1,\ldots,d\}$ and $|u|\leq \lambda_{|J|+1}^{d-|J|}$ otherwise. In both case $|\sigma u|\leq d x^{d^d}$. Since $\lambda$ is an adapted extractor we deduce that $\vec{x}(j)\geq |u|~\norm{\infty}{\vec{A}}$ for every $j\in J$ and by applying Lemma~\ref{lem:holding} we deduce that there exists a run $\vec{x}\xrightarrow{u}\vec{y}$. Since $\pi_J(\vec{x})\xrightarrow{u}\pi_J(\vec{q})$ we deduce that $\vec{y}(j)=\vec{q}(j)$ for every $j\not\in J$. Moreover if $j\in J\moins I$ since $\vec{y}(j)\geq s$ and $s>\norm{\infty}{\vec{q}}$ we get $\vec{y}(j)>\vec{q}(j)$. We deduce that $\vec{q}\leq\vec{y}$ and $J\moins I=\{i \mid \vec{q}(i)\not=\vec{y}(i)\}$. Therefore $\vec{q}$ is forward pumpable by the cycle $\pi_J(\vec{q})\xrightarrow{\sigma u}_{\pi_J(G)}\pi_J(\vec{q})$.

\medskip

Symmetrically we prove the backward case. We have proved Theorem~\ref{thm:pumpable}.

\section{Deciding The Reversibility Problem}\label{sec:deciding}
In this section, the reversible reachability problem is proved to be EXPSPACE-complete. The proof is inspired by the Kosaraju ideas \cite{Kosaraju82}. A word $\alpha\in\vec{A}^*$ is said to be \emph{reversible} on a configuration $\vec{p}$ if there exists a word $\beta\in\vec{A}^*$ such that $\vec{p}\xrightarrow{\alpha \beta}\vec{p}$ and $\Delta(\alpha)+\Delta(\beta)=\vec{0}$. Note that if $\vec{p}$ is a standard configuration the last condition is implied by the first one. 
\begin{thm}\label{thm:revbound}
  Let $\alpha\in\vec{A}^*$ be a reversible word on a configuration $\vec{p}$. There exists another word $\alpha'\in\vec{A}^*$ reversible on $\vec{p}$ such that $\Delta(\alpha)=\Delta(\alpha')$ and such that:
  $$|\alpha'|\leq 17d^2 x^{15d^{d+2}}$$
  where $x=(1+2\norm{\infty}{\vec{A}})~(1+\norm{\infty}{\vec{p}} +\norm{\infty}{\Delta(\alpha)} )$.
\end{thm}

Let us assume that $\alpha\in\vec{A}^*$ is a reversible word on a configuration $\vec{p}\in\setN_I^d$. There exists a run $\vec{p}\xrightarrow{\alpha\beta}\vec{p}$ satisfying $\Delta(\alpha)+\Delta(\beta)=\vec{0}$. From this run we extract a unique witness graph $G=(\vec{Q},T)$ such that $\vec{p}\xrightarrow{\alpha\beta}_G\vec{p}$ is a total cycle. In particular the Parikh image of this cycle is a total Kirchhoff function proving that $G$ is reversible by Lemma~\ref{lem:rev}.

\medskip

We introduce $a=\norm{\infty}{\vec{A}}$ and $s=1+\norm{\infty}{\vec{p}}+\norm{\infty}{\Delta(\alpha)}$. Let $\vec{q}=\vec{p}+\Delta_I(\alpha)$. We have $\norm{\infty}{\vec{q}}\leq \norm{\infty}{\vec{p}}+\norm{\infty}{\Delta(\alpha)}<s$. Let us introduce $x=(1+2a)s$. Theorem~\ref{thm:pumpable} shows that there exists a set of indexes $J$ such that $\pi_J(G)$ has at most $x^{d^d}$ states and such that $\vec{p}$ is forward pumpable by a cycle $\pi_J(\vec{p})\xrightarrow{v}_{\pi_J(G)}\pi_J(\vec{p})$ and $\vec{q}$ is backward pumpable by a cycle $\pi_J(\vec{q})\xrightarrow{w}_{\pi_J(G)}\pi_J(\vec{q})$ such that $|v|,|w|\leq d x^{d^d}$. In particular $\Delta_I(v)$ and $-\Delta_I(w)$ are two vectors in $\{\vec{c}\in\setN_I^d \mid \vec{c}(i)\not=0 \Leftrightarrow i\in J\}$. For every $n\in\setN$ we have:

$$
\vec{p}\xrightarrow{v^n}\vec{p}+n\Delta_I(v) 
~~~~~~
\vec{q}-n\Delta_I(w)\xrightarrow{w^n}\vec{q}
$$

Since the witness graph $G$ is reversible, Lemma~\ref{lem:rev} shows that $\pi_J(G)$ is reversible. From Theorem~\ref{thm:reversible} we deduce that the zero vector is the displacement of a total Kirchhoff function $\mu$ for $\pi_J(G)$ satisfying:
\begin{align*}
  \norm{\infty}{\mu}
  &\leq (x^{d^d}(1+2a))^{d(d+1)}\\
  &\leq (x^{d^d}x)^{2d^2}\\
  &\leq (x^{2d^d})^{2d^2}\\
  &\leq x^{4d^{d+2}}
\end{align*}
Note that $|\pi_J(T)|\leq |\pi_J(\vec{Q})|~|A|\leq x^{d^d} (1+2a)^d\leq x^{2d^d}$.

\begin{lem}
  There exists a cycle $\pi_J(\vec{q})\xrightarrow{u}_{\pi_J(G)}\pi_J(\vec{q})$ such that $\Delta(v)+\Delta(u)+\Delta(w)=\vec{0}$ and:
  $$|u|\leq 3d ~x^{7d^{d+2}}$$
\end{lem}
\proof
  Let $\mu_v,\mu_w$ be the Parikh images of  $\pi_J(\vec{p})\xrightarrow{v}_{\pi_J(G)}\pi_J(\vec{p})$ and
  $\pi_J(\vec{q})\xrightarrow{w}_{\pi_J(G)}\pi_J(\vec{q})$. We introduce the function $\lambda=(1+2d x^{d^d})\mu-(\mu_v+\mu_w)$. Observe that $\lambda$ is a Kirchhoff function for $\pi_J(G)$ satisfying $\lambda(t)\geq (1+2d x^{d^d})-2d x^{d^d}\geq 1$ for every $t\in\pi_J(T)$. The Euler's Lemma shows that $\lambda$ is the Parikh image of a total cycle $\pi_J(\vec{q})\xrightarrow{u}_{\pi_J(G)}\pi_J(\vec{q})$. Observe that $\Delta(u) =(1+2d x^{d^d})\vec{0}-(\Delta(v)+ \Delta(w))$. Hence $\Delta(v)+\Delta(u)+\Delta(w)=\vec{0}$. The length of $u$ is bounded by:
\[
  |u|=\sum_{t\in\pi_J(T)}(1+2d x^{d^d})\mu(t)-(\mu_v(t)+\mu_w(t))
  \leq 3d x^{d^d}  \norm{\infty}{\mu}|\pi_J(T)|
\leq 3d x^{7d^{d+2}}\eqno{\qEd}
\]


\begin{lem}\label{lem:tech}
  There exists a path $\pi_J(\vec{p})\xrightarrow{\tilde{\alpha}}_{\pi_J(G)}\pi_J(\vec{q})$ such that $\Delta(\tilde\alpha)=\Delta(\alpha)$ and:
  $$|\tilde{\alpha}|\leq 2x^{7d^{d+2}}$$
\end{lem}
\proof
  Since $\pi_J(G)$ is strongly connected, there exists a path $\pi_J(\vec{q})\xrightarrow{\tilde{\beta}}_{\pi_J(G)}\pi_J(\vec{p})$. We can assume that $|\tilde{\beta}|$ is minimal. In particular $|\tilde{\beta}|< x^{d^d}$. Moreover, we know that $\pi_J(\vec{p})\xrightarrow{\alpha}_{\pi_J(G)}\pi_J(\vec{q})$. Observe that $\alpha \tilde{\beta}$ is the label of a cycle in $\pi_J(G)$. Hence $\vec{z}=\Delta(\alpha) +\Delta( \tilde{\beta})$ is the displacement of a Kirchhoff function for $G$. We have $\norm{\infty}{\vec{z}}\leq \norm{\infty}{\Delta(\alpha)}+\norm{\infty}{\Delta(\tilde{\beta})}\leq s + |\tilde{\beta}|a$ we get $\norm{\infty}{\vec{z}}\leq s+x^{d^d}a\leq x^{d^d}(1+a)$. Theorem~\ref{thm:bound} shows that $\vec{z}$ is the displacement of a Kirchhoff function $\theta$ for $G$ such that:
\begin{align*}
  \norm{\infty}{\theta}
  &\leq ((x^{d^d})^{d+1}a(1+2a)^d + x^{d^d}(1+a))^d\\
  &\leq (x^{ 2d^{d+1}}a x^d  + x^{d^d}(1+a))^d\\
  &\leq (x^{3d^{d+1}}(1+2a))^d\\
  &\leq (x^{3d^{d+1}}x)^d\\
  &\leq x^{4d^{d+2}}
\end{align*}
We introduce the Parikh image $f$ of the path $\pi_J(\vec{q})\xrightarrow{\tilde{\beta}}_{\pi_J(G)}\pi_J(\vec{p})$. 
Let us add to the strongly connected graph $\pi_J(G)$ an additional transition $t_\bullet$ from $\pi_J(\vec{q})$ to $\pi_J(\vec{p})$ and let $G_\bullet$ be this new graph and $T_\bullet=\pi_J(T)\cup\{t_\bullet\}$ be its set of transitions. Functions $\theta$, $\mu$ and $f$ are extended over $T_\bullet$ by $\theta(t_\bullet)=\mu(t_\bullet)=f(t_\bullet)=0$. We also introduce the Parikh image $f_\bullet$ of $t_\bullet$, i.e. $f_\bullet(t_\bullet)=1$ and $f_\bullet(t)=0$ for every $t\in \pi_J(T)$. Let us observe that $g=\theta+x^{d^d}\mu-f+f_\bullet$ satisfies $g(t)\geq 1$ for every $t\in\pi_J(T)$ since $f(t)<x^{d^d}$. A $g(t_\bullet)=1$ we deduce that $g$ is a Kirchhoff function for $G_\bullet$ satisfying $g(t)\geq 1$ for every $t\in T_\bullet$. The Euler's Lemma shows that $g$ is the Parikh image of a total cycle. Since $g(t_\bullet)=1$ we deduce that $g$ is the Parikh image of a cycle of the form $(\pi_J(\vec{p})\xrightarrow{\tilde{\alpha}}_{\pi_J(G)}\pi_J(\vec{q}))~t_\bullet $. By definition of $g$ we get $\Delta( \tilde{\alpha}) =\vec{z}+x^{d^d}\vec{0}-\Delta(\tilde{\beta})+\vec{0}$. Hence $\Delta(\tilde{\alpha})=\vec{z}-\Delta(\tilde{\beta})$. Since $\vec{z}=\Delta(\alpha)+\Delta(\tilde{\beta})$ we get $\Delta(\tilde{\alpha})=\Delta(\alpha)$. The following inequalities provide the lemma:
\begin{align*}
  |\tilde{\alpha}|&
  \leq |\pi_J(T)|~\norm{\infty}{\theta} +x^{d^d}|\pi_J(T)|~\norm{\infty}{\mu}\\
  &\leq x^{2d^d}(x^{4d^{d+2}}+x^{d^d}x^{4d^{d+2}})\\
  &\leq 2x^{7d^{d+2}}\rlap{\hbox to 248 pt{\hfill\qEd}}
\end{align*}

\begin{lem}
  For every $n\geq |u| a$ we have:
  $$\vec{q}+n\Delta_I(v)\xrightarrow{u^{n}}\vec{q}-n\Delta_I(w)
  $$
\end{lem}
\begin{proof}
  Let $n\geq |u|a$. We introduce the sequence $(\vec{x}_k)_{0\leq k\leq n}$ of configurations $\vec{x}_k=\vec{q}+(n-k)\Delta_I(v)-k\Delta_I(w)$. Since $\pi_J(\vec{x}_k)=\pi_J(\vec{q})$ we deduce that there exists a run from  $\pi_J(\vec{x}_k)$ labelled by $u$. Moreover as $\Delta_I(v)(j)\geq 1$ and $-\Delta_I(w)(j)\geq 1$ for every $j\in J$, we deduce that $\vec{x}_k(j)\geq n\geq |u|a$ for every $j\in J$. Lemma~\ref{lem:holding} shows that there exists a run from $\vec{x}_k$ labelled by $u$. Since $\Delta(v)+\Delta(u)+\Delta(w)=\vec{0}$ we get $\vec{x}_k\xrightarrow{u}\vec{x}_{k+1}$.
\end{proof}

\begin{lem}
  For every $n\geq |\tilde{\alpha}| a$ we have:
  $$
  \vec{p}+n\Delta_I(v)\xrightarrow{\tilde{\alpha}}\vec{q}+n\Delta_I(v)
  $$
\end{lem}
\begin{proof}
  Observe that $\pi_J(\vec{p}+n\Delta_I(v))=\pi_J(\vec{p})$ and $\pi_J(\vec{p})\xrightarrow{\tilde{\alpha}}\pi_J(\vec{q})$. Moreover for every $j\in J$ we have $(\vec{p}+n\Delta_I(v))(j)\geq n\geq |\tilde{\alpha}|a$. From Lemma~\ref{lem:holding} we deduce that there exists a run from $\vec{p}+n\Delta_I(v)$ labelled by $\tilde{\alpha}$. From $\vec{p}\xrightarrow{\alpha}\vec{q}$ we deduce that $\vec{p}+\Delta_I(\alpha)=\vec{q}$. Since $\Delta(\alpha)=\Delta(\tilde{\alpha})$ we deduce that $\vec{p}+\Delta_I(\tilde{\alpha})=\vec{q}$. We deduce the run $\vec{p}+n\Delta_I(v)\xrightarrow{\tilde{\alpha}}\vec{q}+n\Delta_I(v)$.
\end{proof}

Finally, let $n=a\max\{|\tilde{\alpha}|, |u|\}$. We have proved that $\vec{p}\xrightarrow{\alpha'}\vec{q}$ where $\alpha'=v^n \tilde{\alpha} u^n w^n$. Note that $\Delta(\alpha')=\Delta(\alpha)$ since $\Delta(\tilde{\alpha})=\Delta(\alpha)$ and $\Delta(v)+\Delta(u)+\Delta(w)=\vec{0}$. We deduce that $\Delta(\alpha')=\Delta(\alpha)$. As $\vec{q}\xrightarrow{\beta}\vec{p}$ with $\Delta(\alpha)+\Delta(\beta)=\vec{0}$ we deduce that $\alpha'$ is reversible on $\vec{p}$. Note that $n\leq a 3d x^{7d^{d+2}}\leq 3d x^{8d^{d+2}}$. Hence we have:
\begin{align*}
  |\alpha'|
  &\leq   2x^{7d^{d+2}} + 3d x^{8d^{d+2}}(2d x^{d^d}+3d x^{7d^{d+2}}) \\
 &\leq 17d^2 x^{15d^{d+2}}
\end{align*}

We have proved Theorem~\ref{thm:revbound}.
\begin{cor}
  Two standard configurations $\vec{p},\vec{q}$ are in the same strongly connected component of a standard subreachability graph if and only if there exist runs $\vec{p}\xrightarrow{\alpha}\vec{q}$ and  $\vec{q}\xrightarrow{\beta}\vec{p}$ such that:
  $$|\alpha|,|\beta|\leq 17d^2 x^{15d^{d+2}}$$
  where $x=(1+2\norm{\infty}{\vec{A}})(1+2\max\{\norm{\infty}{\vec{p}},\norm{\infty}{\vec{q}}\})$.
\end{cor}

\begin{thm}
  The reversible reachability problem is EXPSPACE-complete.
\end{thm}

\section{Application : Reversibility Domains}\label{sec:application}
During the execution of a VAS some actions are reversible and some not. More precisely, let $\vec{D}_{\vec{a}}$ be the set of standard configurations $\vec{c}$ such that there exists a word $\alpha$ satisfying $\vec{c}\xrightarrow{\vec{a}}\vec{c}+\vec{a}\xrightarrow{\alpha}\vec{c}$. We observe that the set $\vec{D}_{\vec{a}}$ is an upward closed set for the order $\leq$. In fact $\vec{c}\xrightarrow{\vec{a}}\vec{c}+\vec{a}\xrightarrow{\alpha}\vec{c}$ implies the same thing by replacing $\vec{c}$ with a standard configuration $\vec{x}\in\vec{c}+\setN^d$. So $\vec{D}_{\vec{a}}$ is characterized by its finite set of minimal elements $\min(\vec{D}_{\vec{a}})$ for $\leq$. As an application of Theorem~\ref{thm:revbound}, we obtain the following result.
\begin{thm}\label{thm:domain}
  Configurations $\vec{c}\in \min(\vec{D}_{\vec{a}})$ satisfy the following inequality where $a=\norm{\infty}{\vec{A}}$.
  $$\norm{\infty}{\vec{c}}\leq (102d^2a^2)^{(15d^{d+2})^{d+2}}$$
\end{thm}
\begin{proof}
  Observe that if $a=0$ we are done since in this case $\vec{c}=\vec{0}$. So we can assume that $a\geq 1$. We introduce the extractor $\lambda=(\lambda_1,\ldots,\lambda_d)$ defined by $\lambda_{d+1}=a$ and the following induction for $n\in\{1,\ldots,d+1\}$:
  $$\lambda_{n-1}=17d^2(6a\lambda_n)^{15d^{d+2}}$$
  Let $\vec{c}\in \min (\vec{D}_{\vec{a}})$ and let $\vec{d}=\vec{c}+\vec{a}$. Let us consider the minimal excluding set $I$ for $(\lambda,\{\vec{d}\})$. By minimality of $I$ we have $\vec{d}(i)<\lambda_{|I|+1}$ for every $i\not\in I$ and $\vec{d}(i)\geq \lambda_{|I|}$ for every $i\in I$. We consider the standard configuration $\vec{y}$ defined by $\vec{y}(i)=\lambda_{|I|}$ if $i\in I$ and $\vec{y}(i)=\vec{d}(i)$ if $i\not\in I$. Let us consider $\vec{q}=\pi_I(\vec{c})$ and $\vec{p}=\pi_I(\vec{d})$. Since $\vec{c}\in\vec{D}_{\vec{a}}$ there exists a run $\vec{d}\xrightarrow{\alpha}\vec{c}$. In particular $\vec{p}\xrightarrow{\alpha}\vec{q}\xrightarrow{\vec{a}}\vec{p}$ with $\Delta(\alpha)+\Delta(\vec{a})=\vec{0}$. We deduce that $\alpha$ is reversible on $\vec{p}$ and Theorem~\ref{thm:revbound} shows that there exists a word $\alpha'$ such that  $\vec{p}\xrightarrow{\alpha'}\vec{q}$, $\Delta(\alpha')=\Delta(\alpha)$ and:
  $$|\alpha'|\leq 17 d^2x^{15d^{d+2}}$$
  where $x=(1+2a)(1+\norm{\infty}{\vec{p}}+\norm{\infty}{\vec{a}})$. Note that $\norm{\infty}{\vec{p}}\leq \lambda_{|I|+1}-1$. We deduce that $x\leq (1+2a)(\lambda_{|I|+1}+a)\leq 6a \lambda_{|I|+1}$ since $1\leq a$ and $a\leq \lambda_{|I|+1}$. Hence $a|\alpha'|\leq \lambda_{|I|}$ thanks to the induction defining $\lambda$. Since $\pi_I(\vec{y})=\vec{p}$ we deduce that there exists a run from $\pi_I(\vec{y})$ labelled by $\alpha'$. As $\vec{y}(i)\geq \lambda_{|I|}\geq a|\alpha'|$ for every $i\in I$, Lemma~\ref{lem:holding} shows that there exists a run $\vec{y}\xrightarrow{\alpha'}\vec{x}$. Since $\Delta(\alpha')=\Delta(\alpha)=-\vec{a}$ we deduce that $\vec{x}=\vec{y}-\vec{a}$. From $\vec{y}\leq \vec{d}$ we get $\vec{x}\leq \vec{c}$ by subtracting $\vec{a}$. Moreover as $\vec{x}\xrightarrow{\vec{a}}\vec{y}\xrightarrow{\alpha'}\vec{x}$ we deduce that $\vec{x}\in\vec{D}_{\vec{a}}$. By minimality of $\vec{c}$ we get $\vec{c}=\vec{x}$. Hence $\vec{c}=\vec{y}-\vec{a}$. In particular $\norm{\infty}{\vec{c}}\leq \lambda_{|I|}+a\leq \lambda_0+a$. Finally let us get a bound on $\lambda_0$. We get the equality $\lambda_{n-1}=c\lambda_n^e$ by introducing $e=15d^{d+2}$ and $c=17d^2(6a)^e$. Hence $\lambda_0\leq (ca)^{e^{d+1}}\leq (102d^2a^2)^{ e^{d+2}}$ and from $e^{d+2}\leq (15d^{d+2})^{d+2}$ we are done.
\end{proof}

\newcommand{\transformer}[1]{\stackrel{#1}{\curvearrowright}}
\newcommand{\fo}[1]{\operatorname{FO}\left(#1\right)}
\section*{Conclusion}
The reversible reachability problem is proved to be EXPSPACE-complete in this paper. The proof is inspired by the Rackoff and Kosaraju ideas \cite{rackoff78,Kosaraju82}. We have introduced the domain of reversibility $\vec{D}_{\vec{a}}$ of every action $\vec{a}\in \vec{A}$. Observe that the reflexive and transitive closure of the following relation $R$ is a congruence and from \cite{Bouziane97} we deduce that this relation is definable in the Presburger arithmetic. That means there exist a Presburger formula $\phi$ that exactly denotes the pair $(\vec{x},\vec{y})$ of standard configurations in the reversible reachability relation. As a future work we are interested in characterizing precisely the size of such a formula (we already derive an elementary bound from \cite{Bouziane97} and Theorem~\ref{thm:domain}). 
$$R=\bigcup_{\vec{a}\in\vec{A}}\{(\vec{x},\vec{x}+\vec{a}) \mid \vec{x}\in\vec{D}_{\vec{a}}\}$$

The general vector addition system reachability problem was recently proved to be decidable thanks to inductive invariants definable in the Presburger arithmetic~\cite{L-popl11}. The proof is based on binary relations called \emph{transformer relations} over $\setQ_{\geq 0}^d$ where $\setQ_{\geq 0}$ is the set of non-negative rational numbers. The \emph{transformer relation} of a standard configuration $\vec{c}\in\setN^d$ is the binary relation $\transformer{\vec{c}}$ over $\setQ_{\geq 0}^d$ defined by $\vec{x}\transformer{\vec{c}}\vec{y}$ if there exists a run from $\vec{c}+n\vec{x}$ to $\vec{c}+n\vec{y}$ for some $n\in\setN$. This relation is proved to be definable in $\fo{\setQ,+,\leq}$ in \cite{L-popl11}. The proof is based on witness graphs. However, no upper bound on the size of these graphs are derived. As a future work we are interested in adapting techniques introduced in this paper for proving elementary upper-bounds on sizes of formulas in $\fo{\setQ,+,\leq}$ denoting transformer relations. Our main objective is the complexity of the general vector addition system reachability problem.

  \bibliography{thisbiblio}
  \bibliographystyle{alpha}

\end{document}